\documentclass[submission,copyright,creativecommons]{eptcs}
\providecommand{\event}{FROM 2026}

\usepackage{iftex}

\ifpdf
  \usepackage{underscore} 
  \usepackage[T1]{fontenc}     
\else
  \usepackage{breakurl}
\fi
\usepackage[utf8]{inputenc}
\usepackage{amsmath,amssymb,amsthm,mathtools}
\usepackage{enumitem}
\usepackage{url}

\title{Subsumption-Free Private-Pivot Learning in Resolvable Network-Based SAT Solving}
\author{G\'abor Kusper
\institute{Faculty of Informatics, Eszterh\'azy K\'aroly Catholic University, Eger, Hungary}
\email{kusper.gabor@uni-eszterhazy.hu}
}

\def\titlerunning{Subsumption-Free Private-Pivot Learning}
\def\authorrunning{G. Kusper}

\setlist{nosep,leftmargin=*}

\newtheorem{definition}{Definition}[section]
\newtheorem{lemma}[definition]{Lemma}
\newtheorem{theorem}[definition]{Theorem}
\newtheorem{corollary}[definition]{Corollary}

\newtheorem{example}[definition]{Example}

\newcommand{\Source}{\mathrm{Source}}
\newcommand{\Sink}{\mathrm{Sink}}
\newcommand{\SN}{\mathit{SN}}
\newcommand{\RN}{\mathcal{N}}
\newcommand{\R}{\mathcal{R}}
\newcommand{\DNF}{\mathrm{DNF}}
\newcommand{\W}{\mathcal{W}}
\newcommand{\T}{\mathcal{T}}
\newcommand{\new}{\mathrm{new}}

\begin{document}
\maketitle

\begin{abstract}
A resolvable network is a directed-graph representation of SAT: every SAT instance can be translated into an RN, and every RN has an associated CNF formula. Each reach represents one clause. In a mixed reach, the head and tail are disjoint sets of variables containing the variables occurring negatively and positively in the clause, respectively; the distinguished symbols Source and Sink represent a missing negative- or positive-literal side. RN-Solver is a proof-of-concept SAT solver based on this representation. Its all-positive clauses are represented by white reaches, and its token distributions are the inclusion-minimal hitting sets of the current white tails, generated by monotone CNF--DNF dualization.
RN-Solver learns new white reaches by private-pivot resolution, a structured resolution sequence that uses old white reaches as pivot witnesses. In the original algorithm, every candidate white reach generated by such a chain was followed by a global subsumption test against the current network. Profiling showed that this subsumption test can dominate the runtime on random 3-SAT instances. We show that this check is unnecessary when the mixed reach used for learning is falsified by the current token distribution, meaning that the distribution makes all variables in the head true and all variables in the tail false.
The key invariant is simple: the resulting white tail is disjoint from the triggering token distribution, while the same distribution intersects every old white tail. Hence no old white reach can subsume the generated reach. This structural observation allows us to construct a simpler, snapshot-based, subsumption-free variant of RN-Solver, where each main-loop iteration uses a fixed set of old reaches and installs newly generated white reaches only at the end of the iteration. We prove soundness of the revised algorithm. Empirical evaluation confirms the elimination of the targeted checks: within an 8 second budget, snapshot/full-DNF solves 963 of 1,000 uf20-91 instances, compared with 724 for the original control flow. A separate 100-instance comparison identifies exact incremental DNF as the most effective of the three tested configurations.
\end{abstract}

\section{Introduction}

Resolvable networks (RNs) give a directed-graph view of propositional CNF formulas. A mixed reach is an ordered pair \((A,B)\) of disjoint non-empty sets of variables, called its head and tail. It represents the clause
\[
   \bigvee_{a\in A}\neg a \;\vee\; \bigvee_{b\in B} b,
\]
or equivalently the implication saying that if all variables in \(A\) are true, then at least one variable in \(B\) must be true. Thus the head contains the variables occurring negatively in the clause, and the tail contains the variables occurring positively. The distinguished symbols \(\Source\) and \(\Sink\) are used when the negative- or positive-literal side is absent, respectively. Every CNF clause can be represented as a reach, and every RN has an associated CNF formula; RNs are therefore not a new logic, but a structured graph representation of SAT.

The motivation for studying RN-based SAT solving is not to replace mature CDCL solvers on raw performance at its present proof-of-concept stage. Our question is whether structural properties can simplify reasoning within RN-Solver, not whether its implementation is already competitive with highly optimized SAT solvers. RN-Solver explores a different structural organization of SAT reasoning, combining a clause-level graph representation, monotone CNF--DNF dualization, and restricted resolution. Pure resolution may generate too many resolvents, while full dualization may generate too many minimal transversals; RN-Solver studies how these mechanisms can guide one another. In particular, this organization can expose invariants that turn runtime checks into theorem-level guarantees. The present paper gives a concrete example: the minimal-hitting-set structure that supplies private pivots also certifies that the learned white reach cannot be subsumed by any old white reach. The interest is therefore conceptual and algorithmic, while substantial scope remains for further performance improvements.

The special reaches determine the terminology used by RN-Solver. A reach of the form \((\Source,B)\) represents an all-positive clause and is called a white reach; the set \(B\) is its white tail. A reach of the form \((A,\Sink)\) represents an all-negative clause and is called a black reach. A reach with both a non-empty head and a non-empty tail is mixed. White reaches are the source of the token distributions used by RN-Solver: a token distribution is a set of variables interpreted as true, and it must intersect every current white tail. In other words, the token distributions generated by RN-Solver are the inclusion-minimal hitting sets of the current white tails, computed by monotone CNF--DNF dualization.

The original RN-Solver introduced in~\cite{KusperNagy2026RNSolver} uses these token distributions as a global reasoning device, instead of branching on assignments in the DPLL/CDCL style. Once a token distribution has been generated, the solver checks how it interacts with the non-white reaches. Informally, a token distribution falsifies a mixed reach if it makes all variables in the head true and all variables in the tail false. In that case, the mixed reach exposes a conflict with the current distribution. RN-Solver responds by a private-pivot chain: a structured sequence of resolution steps that uses suitable old white reaches as pivot witnesses and derives a new white reach.

This paper refines this private-pivot learning step. In the original RN-Solver, every newly generated candidate \((\Source,B_{\new})\) is tested for novelty and non-subsumption before being added to the network. 
For white reaches, this is ordinary clause subsumption: 
an old white reach \((\Source,B_i)\) subsumes \((\Source,B_{\new})\) when \(B_i\subseteq B_{\new}\). 
The profiling reported in~\cite{KusperNagy2026RNSolver} suggests that this check can dominate the runtime on random 3-SAT instances. 
Our main observation is structural: when the mixed reach used for learning is falsified by the current token distribution, the private-pivot chain produces a new white tail outside that distribution. Since the same distribution intersects every old white tail, no old white reach can subsume the generated one. This structural observation allows us to construct a simpler, snapshot-based, subsumption-free variant of RN-Solver.
\begin{example}[The key idea in miniature]
Let the old white tails be \(B_1=\{1,2\}\) and \(B_2=\{3,4\}\), and let \(C=\{1,3\}\), which hits both tails. For the \(C\)-falsified mixed reach \((\{1,3\},\{5\})\), the two private witnesses are \(B_1\) and \(B_2\), and the private-pivot construction gives
\[
B_{\new}=(B_1\cup B_2\cup\{5\})\setminus\{1,3\}=\{2,4,5\}.
\]
Thus \(B_{\new}\cap C=\emptyset\), while each old white tail intersects \(C\). Hence neither old white tail can be contained in \(B_{\new}\), so the new white reach cannot be subsumed by an old one.
\end{example}

We use a frozen snapshot: the white, mixed, and black reaches used during one iteration of the solver's main loop are fixed at the beginning of that iteration.
The original algorithm mutates the network while processing token distributions generated from an earlier white family. 
In the revised solver, each main-loop iteration starts with a snapshot \(\RN_t\), computes \(\T_t=\DNF(\W_t)\), collects new white reaches in a batch \(\Delta_t\), and installs the batch only after the iteration. This makes clear which reaches are old and which are newly generated. It also separates logical soundness, strict progress, and possible redundancy inside the new batch.
This snapshot discipline has a deliberate trade-off: a reach learned early in an iteration cannot affect token distributions processed later in that same iteration, so it may increase within-iteration work and permit batch-local redundancy; in return, it provides the clean old/new boundary on which the old-white non-subsumption guarantee relies.

The resulting control flow is simpler. A token distribution may be terminal, in which case it gives a SAT witness; it may be black-blocked, in which case it is already forced into contradiction by a black reach; or it may be active, in which case a learning step is possible. In the original solver, active token distributions were further divided into resolved and non-resolved cases, leading to a white-fixpoint SAT branch. In the snapshot-based formulation of this paper, every active token distribution is productive: it yields a new white reach that is not subsumed by the old white reaches. Thus the resolved branch and the white-fixpoint SAT branch disappear from the main control flow.

The main contributions are: (i) the definition of \(C\)-falsified reaches, formalizing the falsification condition described above; (ii) the C-Falsified Private-Pivot Complement Lemma; (iii) the Subsumption-Free Private-Pivot Lemma; (iv) strict white refinement and token-distribution elimination; (v) a revised snapshot/full-DNF algorithm without the resolved and white-fixpoint branches; and (vi) empirical evaluation of the revised control flow, including validation on all 1,000 \texttt{uf20-91} instances and a runtime--memory comparison of three implemented DNF-control configurations.

\section{Related Work}

\paragraph{SAT solving: resolution, splitting, and CDCL.}
The two classical algorithmic lines for SAT already appear in the earliest literature. The Davis--Putnam procedure~\cite{DavisPutnam1960} is resolution-based: variables are eliminated by deriving resolvents. The later Davis--Logemann--Loveland procedure~\cite{DavisLogemannLoveland1962} replaces exhaustive elimination by splitting and backtracking. Modern CDCL solvers extend this search-based line with unit propagation, implication graphs, conflict analysis, clause learning, restarts, and inprocessing; early milestones include GRASP~\cite{MarquesSilvaSakallah1999}, Chaff~\cite{Moskewicz2001Chaff}, MiniSat~\cite{EenSorensson2003MiniSat}, and Glucose-style learnt-clause management~\cite{AudemardSimon2009Glucose}. The broader SAT background is covered in the \emph{Handbook of Satisfiability}~\cite{Biere2021Handbook}. RN-Solver belongs to the resolution-based line but uses a graph-guided, dualization-driven control structure rather than DPLL/CDCL-style branching.

\paragraph{Resolution and restricted resolution strategies.}
Resolution itself is a foundational inference rule for automated reasoning: Robinson's resolution principle~\cite{Robinson1965Resolution} provides a sound and refutation-complete calculus for clauses. Many restricted proof procedures were developed to control the search space. Linear and input-style resolution, including linear resolution with selection functions~\cite{KowalskiKuehner1971LinearResolution}, model elimination~\cite{Loveland1968ModelElimination}, and connection-graph procedures~\cite{Kowalski1975ConnectionGraphs}, constrain the shape or scheduling of resolution derivations. RN-Solver is related to this family only at the level of proof shape. A private-pivot chain resolves one selected mixed reach against a sequence of selected side clauses, but those side clauses are current old white reaches, not necessarily original input clauses, and the pivot sequence is supplied by the private-element property of a minimal hitting set. Thus RN-Solver is best viewed as a dualization-guided restricted resolution method, not as an instance of input resolution.

\paragraph{Resolution proof complexity.}
The Cook--Reckhow framework~\cite{CookReckhow1979ProofSystems} formalizes propositional proof systems. Classical results include Haken's exponential lower bound for pigeonhole formulas~\cite{Haken1985Resolution}, hard formula families of Chv\'atal and Szemer\'edi~\cite{ChvatalSzemeredi1988HardResolution}, and the size--width relation of Ben-Sasson and Wigderson~\cite{BenSassonWigderson2001}. RN-Solver remains subject to resolution-style worst-case limitations; its interest is how the RN representation and dualization layer guide restricted resolution, not how to evade those limitations.

\paragraph{Monotone dualization and hypergraph transversals.}
The second non-standard ingredient of RN-Solver is monotone CNF--DNF dualization. The token distributions used by RN-Solver are the inclusion-minimal hitting sets of the current white tails, equivalently the minimal transversals of a hypergraph. The complexity of monotone dualization was studied by Fredman and Khachiyan~\cite{FredmanKhachiyan1996}, and related formulations were investigated by Eiter and Gottlob~\cite{EiterGottlob1995} and Bioch and Ibaraki~\cite{BiochIbaraki1995}. Practical transversal enumeration algorithms include MMCS-style methods~\cite{MurakamiUno2014}, and computational comparisons are given by Gainer-Dewar and Vera-Licona~\cite{GainerDewarVeraLicona2017}; Berge's monograph remains a standard reference for hypergraphs~\cite{Berge1984}. Here transversals guide resolution and certify non-subsumption; we do not propose a new enumerator.

\paragraph{Graph-based and compiled representations.}
BDDs canonically represent Boolean functions for a fixed variable ordering~\cite{Bryant1986BDD}, whereas ZDDs target sparse set families~\cite{Minato1993ZDD}. Knowledge compilation studies languages such as DNNF that support tractable queries~\cite{DarwicheMarquis2002KC}. Unlike these function- or solution-space representations, an RN retains a clause-level correspondence and expresses resolution as an operation on reaches.

\paragraph{Resolvable networks at the intersection of resolution and dualization.}
Resolvable networks were introduced as a graph-based representation of SAT in which clauses correspond to directed reaches~\cite{KusperBiroNagy2021RN}. The later RN-Solver paper~\cite{KusperNagy2026RNSolver} used this representation to combine monotone CNF--DNF dualization with private-pivot resolution. In this framework, transversals identify token distributions, token distributions identify private pivots, and private pivots generate new white reaches.

\paragraph{Subsumption, redundancy, and saturation.}
Subsumption and redundancy elimination are standard tools for controlling generated clauses. Saturation-style theorem provers derive clauses by inference rules and use clause selection, simplification, and redundancy criteria to control growth; see the resolution-theorem-proving chapter of the \emph{Handbook of Automated Reasoning}~\cite{BachmairGanzinger2001Resolution} and the E prover system description~\cite{Schulz2013E}. In SAT solving, preprocessing and inprocessing techniques such as variable elimination and clause elimination are central engineering tools~\cite{EenBiere2005Preprocessing}; blocked-clause elimination is a prominent resolution-environment-based redundancy criterion~\cite{JarvisaloBiereHeule2010Blocked}. In contrast, our result rules out old-white subsumption for private-pivot candidates generated from falsified mixed reaches, allowing the corresponding runtime test to be removed.

\paragraph{Snapshot discipline.}
The snapshot discipline separates old and newly generated clauses, as in saturation procedures, but specifically fixes the white family used both to generate token distributions and to prove old-white non-subsumption. Batch-local redundancy remains a separate implementation issue.

\paragraph{Gap addressed by this paper.}
RN-Solver links redundancy control in resolution to minimal hitting sets. Our contribution is not a new transversal enumerator or subsumption data structure: when the selected mixed reach is falsified, the dualization witness itself certifies non-subsumption of the private-pivot candidate, replacing a runtime test with a structural guarantee.

\section{Definitions}
\label{sec:definitions}

We use only the RN notation needed in this paper; all standard definitions are as in~\cite{KusperNagy2026RNSolver}. Let \(V\) be a finite set of Boolean variables. A resolvable network over \(V\) is the directed graph
\[
\RN=(\SN\cup\{\Source,\Sink\},\R),
\]
where \(\SN\subseteq\mathcal{P}(V)\setminus\{\emptyset\}\) is a family of ordinary subnetworks, each being a non-empty set of variables, and \(\Source\) and \(\Sink\) are distinguished symbols, distinct from each other and not belonging to \(\SN\). Let
\[
   U=\bigcup\SN
\]
be the universe of variables occurring in the ordinary subnetworks. The reach set \(\R\) contains only the following four forms:
\begin{itemize}
\item a mixed reach \((A,B)\), where \(A,B\in\SN\) and \(A\cap B=\emptyset\);
\item a white reach \((\Source,B)\), where \(B\in\SN\);
\item a black reach \((A,\Sink)\), where \(A\in\SN\); or
\item the empty reach \((\Source,\Sink)\).
\end{itemize}
The first component of a reach is its head and the second is its tail. No other reach forms are allowed: \(\Source\) may occur only as the first component (head), and \(\Sink\) only as the second component (tail). The corresponding clauses are
\[
\begin{aligned}
(A,B) &\longmapsto \bigvee_{a\in A}\neg a \;\vee\; \bigvee_{b\in B} b,\\
(\Source,B) &\longmapsto \bigvee_{b\in B} b,\\
(A,\Sink) &\longmapsto \bigvee_{a\in A}\neg a,\\
(\Source,\Sink) &\longmapsto \bot.
\end{aligned}
\]
Thus \(\Source\) represents the absence of a negative-literal side, while \(\Sink\) represents the absence of a positive-literal side.

Let the current white reaches be \((\Source,B_1),\ldots,(\Source,B_k)\). Their token distributions are the minimal hitting sets
\[
\DNF(B_1,\ldots,B_k)=\{C\subseteq U\mid C\cap B_i\ne\emptyset\text{ for all }i,
\text{ and }C\text{ is inclusion-minimal}\}.
\]
We use the standard private-element property: if \(C\in\DNF(B_1,\ldots,B_k)\), then for every \(c\in C\) there is an index \(j(c)\) with \(B_{j(c)}\cap C=\{c\}\).

A token distribution \(C\) is black-blocked if some black reach \((G,\Sink)\) has \(G\subseteq C\). It is terminal if the valuation clause represented by \((C,U\setminus C)\) is not subsumed by any reach in the network. If neither case holds, \(C\) is active.

\begin{definition}[\(C\)-falsified reach]
Let \(C\subseteq U\), and let the associated valuation make exactly the variables in \(C\) true. A reach is \emph{\(C\)-falsified} if its corresponding clause is false under this valuation. Equivalently,
\begin{itemize}
\item a mixed reach \((G,H)\) is \(C\)-falsified iff \(G\subseteq C\) and \(H\cap C=\emptyset\);
\item a white reach \((\Source,H)\) is \(C\)-falsified iff \(H\cap C=\emptyset\);
\item a black reach \((G,\Sink)\) is \(C\)-falsified iff \(G\subseteq C\); and
\item the empty reach \((\Source,\Sink)\) is \(C\)-falsified for every \(C\).
\end{itemize}
\end{definition}

For a \(C\)-falsified mixed reach \((G,H)\), choose for each \(g\in G\) a private witness \(B_{j(g)}\) with \(B_{j(g)}\cap C=\{g\}\). The private-pivot white tail is
\begin{equation}
\label{eq:bnew}
    B_{\new}=\left(\left(\bigcup_{g\in G}B_{j(g)}\right)\cup H\right)\setminus G.
\end{equation}
An iteration snapshot, or snapshot, is a pair \((\RN_t,\W_t)\) fixed at the beginning of a main-loop iteration, after the initial unit-propagation and decision checks; in particular, \(\RN_t\) contains no empty reach \((\Source,\Sink)\). The DNF \(\T_t=\DNF(\W_t)\) and all old-white non-subsumption statements are interpreted with respect to this snapshot.

\section{Theoretical Results}
\label{sec:theory}

\begin{lemma}[No old white reach is \(C\)-falsified]
Let \(C\in\DNF(B_1,\ldots,B_k)\). Then no current white reach \((\Source,B_i)\) is \(C\)-falsified.
\end{lemma}
\begin{proof}
Every \(B_i\) intersects \(C\), whereas a white reach \((\Source,B_i)\) is \(C\)-falsified exactly when \(B_i\cap C=\emptyset\).
\end{proof}

\begin{lemma}[Mixed witness]
Let \(C\in\DNF(\W_t)\). If \(C\) is neither terminal nor black-blocked in the snapshot \(\RN_t\), then \(\RN_t\) contains a \(C\)-falsified mixed reach.
\end{lemma}
\begin{proof}
Since \(C\) is not terminal, the valuation clause represented by \((C,U\setminus C)\) is subsumed by some reach of \(\RN_t\); by the reach semantics, that reach is \(C\)-falsified. It is not black, because \(C\) is not black-blocked, and it is not white by the previous lemma. It is not the empty reach by the snapshot convention. Hence it is mixed.
\end{proof}

\begin{lemma}[C-Falsified Private-Pivot Complement]
Let \(C\in\DNF(B_1,\ldots,B_k)\), and let \((G,H)\) be a \(C\)-falsified mixed reach. Let \(B_{\new}\) be defined by~\eqref{eq:bnew}. Then
\[
    B_{\new}\subseteq U\setminus C.
\]
\end{lemma}
\begin{proof}
If \(x\in H\), then \(x\notin C\) by \(C\)-falsification. If \(x\in B_{j(g)}\setminus\{g\}\), then \(x\notin C\) because \(B_{j(g)}\cap C=\{g\}\). Since all pivots \(g\in G\) are removed in~\eqref{eq:bnew}, every element of \(B_{\new}\) lies outside \(C\).
\end{proof}

\begin{theorem}[Subsumption-Free Private-Pivot]
Under the assumptions of the previous lemma, no current old white reach \((\Source,B_i)\) subsumes \((\Source,B_{\new})\).
\end{theorem}
\begin{proof}
Every old white tail \(B_i\) intersects \(C\), while \(B_{\new}\cap C=\emptyset\). Thus \(B_i\subseteq B_{\new}\) is impossible.
\end{proof}

\begin{example}[A subsumption-free private-pivot step]
Let the current white tails be
\[
    B_1=\{1,2\},\qquad B_2=\{3,4\},\qquad B_3=\{5,6\}.
\]
One token distribution is
\[
    C=\{1,3,5\}\in\DNF(B_1,B_2,B_3).
\]
Consider the mixed reach
\[
    (G,H)=(\{1,3\},\{7,8\}).
\]
It is falsified by \(C\), since \(G\subseteq C\) and \(H\cap C=\emptyset\).
The private witnesses for the pivots are \(B_1\) and \(B_2\), because
\[
    B_1\cap C=\{1\},\qquad B_2\cap C=\{3\}.
\]
The private-pivot construction gives
\[
    B_{\new}
    =
    (B_1\cup B_2\cup H)\setminus G
    =
    (\{1,2\}\cup\{3,4\}\cup\{7,8\})\setminus\{1,3\}
    =
    \{2,4,7,8\}.
\]
Thus \(B_{\new}\cap C=\emptyset\). On the other hand, every old white tail intersects \(C\):
\[
    B_1\cap C=\{1\},\qquad
    B_2\cap C=\{3\},\qquad
    B_3\cap C=\{5\}.
\]
Therefore no old white tail can be contained in \(B_{\new}\), and the generated white reach \((\Source,B_{\new})\) cannot be subsumed by any old white reach.
\end{example}

\begin{example}[Why the falsification condition is needed]
The condition \(H\cap C=\emptyset\) cannot be dropped. Let
\[
    B_1=\{1,2\},\qquad B_2=\{3,4\},
\]
and take the token distribution
\[
    C=\{1,3\}.
\]
Now consider the mixed reach
\[
    (G,H)=(\{1\},\{3,4\}).
\]
Here \(G\subseteq C\), but the reach is not falsified by \(C\), because
\[
    H\cap C=\{3\}\neq\emptyset.
\]
If we nevertheless perform the private-pivot construction using the witness \(B_1\), we obtain
\[
    B_{\new}
    =
    (B_1\cup H)\setminus\{1\}
    =
    (\{1,2\}\cup\{3,4\})\setminus\{1\}
    =
    \{2,3,4\}.
\]
But now the old white tail \(B_2=\{3,4\}\) satisfies
\[
    B_2\subseteq B_{\new}.
\]
Thus the generated white reach would be subsumed by an old white reach. This shows why the revised solver learns subsumption-free candidates only from mixed reaches falsified by the current token distribution.
\end{example}

\begin{lemma}[Private-pivot soundness without subsumption]
Every white reach obtained by a private-pivot chain from reaches of \(\RN_t\) is a logical consequence of \(\RN_t\). Adding it preserves satisfiability, regardless of whether it is subsumed.
\end{lemma}
\begin{proof}
A private-pivot chain is a finite sequence of reach-resolutions, and reach-resolution is ordinary clause resolution in RN notation. Resolution is sound.
\end{proof}

\begin{theorem}[Strict White Refinement]
Let \(F_{\W_t}=\bigwedge_i(\bigvee_{b\in B_i}b)\) be the old white abstraction. If \(B_{\new}\) is generated from a \(C\)-falsified mixed reach in the snapshot \(\RN_t\), then
\[
    \RN_t^{\wedge}\models \bigvee_{b\in B_{\new}} b,
    \qquad\text{but}\qquad
    F_{\W_t}\not\models \bigvee_{b\in B_{\new}} b.
\]
\end{theorem}
\begin{proof}
The first statement is private-pivot soundness. For the second, the valuation \(V_C\) satisfies all old white clauses because \(C\) hits every \(B_i\). By the complement lemma, \(B_{\new}\cap C=\emptyset\), so \(V_C\) falsifies the new white clause.
\end{proof}

\begin{corollary}[Token-distribution elimination]
After adding \((\Source,B_{\new})\), the generating token distribution \(C\) is not in \(\DNF(\W_t\cup\{B_{\new}\})\).
\end{corollary}
\begin{proof}
The set \(C\) does not hit \(B_{\new}\).
\end{proof}
\begin{example}[Eliminating the generating distribution]
In the introductory example, adding the new white tail \(B_{\new}=\{2,4,5\}\) eliminates \(C=\{1,3\}\) immediately: \(C\) still hits the two old white tails, but it does not hit \(B_{\new}\). Thus \(C\) cannot occur in the next DNF of the white family.
\end{example}

\begin{theorem}[Snapshot active distributions are productive]
In an iteration snapshot, every active \(C\in\DNF(\W_t)\) yields a \(C\)-falsified mixed reach and a private-pivot white reach not subsumed by old white reaches. Hence there is no active-but-resolved token distribution in the snapshot sense.
\end{theorem}
\begin{proof}
By the mixed-witness lemma, an active \(C\) has a \(C\)-falsified mixed witness. The Subsumption-Free Private-Pivot Theorem gives a resolvent that is not subsumed by old white reaches.
\end{proof}

\begin{theorem}[Snapshot decide-or-refine]
Let \(\T_t=\DNF(\W_t)\). One snapshot-based iteration of the revised solver satisfies exactly one of the following outcomes:
some \(C\in\T_t\) is terminal and gives SAT; every \(C\in\T_t\) is black-blocked and gives UNSAT; or at least one new old-white-non-subsumed white reach can be generated.
\end{theorem}
\begin{proof}
If neither decision case holds, some \(C\in\T_t\) is neither terminal nor black-blocked. It is active and therefore productive by the previous theorem.
\end{proof}

\paragraph{Unit generation.}
Let \(C\) be a token distribution and \((G,H)\) a
\(C\)-falsified mixed reach. For a fixed choice of private
witnesses \(B_{j(g)}\) satisfying
\(B_{j(g)}\cap C=\{g\}\) for every \(g\in G\),
the private-pivot tail can be written as
\[
    B_{\new}
    = H\cup\bigcup_{g\in G}
      \bigl(B_{j(g)}\setminus\{g\}\bigr).
\]
Since \(H\neq\varnothing\) for a mixed reach,
\(B_{\new}=\{u\}\) holds exactly when \(H=\{u\}\) and
\(B_{j(g)}\setminus\{g\}\subseteq\{u\}\) for every \(g\in G\).
In particular, suppose that a mixed reach \((G,\{u\})\)
is present and that, for every \(g\in G\), the white reach
\((\Source,\{g,u\})\) is also present.
For any token distribution \(C\) with \(G\subseteq C\)
and \(u\notin C\), the tails of these white reaches are
valid private witnesses, because
\(\{g,u\}\cap C=\{g\}\).
Choosing \(B_{j(g)}=\{g,u\}\) for every \(g\in G\)
therefore yields the unit \((\Source,\{u\})\).
Other valid choices of private witnesses need not yield a unit.
This explains the unit emergence pattern of
Example~4 in~\cite{KusperNagy2026RNSolver}.
We record it here only as a structural observation;
the reference algorithm below does not use a special rule
for this configuration.

\paragraph{Exploratory DNF-control observations.}
If \(C\in\DNF(\W)\), \(B_{\new}\cap C=\emptyset\), and \(b\in B_{\new}\), then \(C\cup\{b\}\) is a minimal hitting set of \(\W\cup\{B_{\new}\}\) exactly when every old element \(x\in C\) has a private witness avoiding \(b\). This gives the lost-private heuristic used in the exploratory implementation. Conversely, a partial list of DNF terms cannot in general be updated into the new DNF after adding a white tail: for \(\W=\{\{a,b\}\}\), knowing only the partial term \(\{a\}\) is insufficient after adding \(\{b,c\}\), because the true new DNF is \(\{\{b\},\{a,c\}\}\). Hence exact incremental DNF update requires the full old DNF or an independent completeness certificate.

\section{From the Original Algorithm to the Revised Algorithm}
\label{sec:comparison}

The original RN-Solver processes all token distributions generated from the current white reaches. For a non-black-blocked token distribution \(C\), it considers reaches \((C',D)\) with \(C'\subseteq C\) and \(D\ne\Sink\), constructs a private-pivot candidate, and explicitly checks whether this candidate is new and not subsumed. If no new candidate is found, \(C\) is marked as resolved. If all token distributions are black-blocked or resolved, the old algorithm returns SAT via the white-fixpoint theorem.

The revised algorithm changes this control flow in three ways. First, it uses a frozen iteration snapshot and adds a batch of new white reaches only after the iteration. Second, it learns only from \(C\)-falsified mixed reaches. Third, for those candidates it does not perform old-white subsumption testing; non-subsumption is guaranteed by the theorem above. Consequently, the resolved marker and the white-fixpoint SAT branch are not part of the revised solver.

The simplification can be summarized as follows.
\[
\begin{array}{lll}
\textbf{Old control flow} & & \textbf{Revised control flow} \\
\hline
\text{generate all }C\in\DNF(\W) && \text{generate all }C\in\DNF(\W_t)\text{ from a snapshot}\\
\text{try all }C'\subseteq C\text{ with }D\ne\Sink && \text{use only }C\text{-falsified mixed reaches}\\
\text{test candidate for old-white subsumption} && \text{skip this test by theorem}\\
\text{mark active terms as resolved/non-resolved} && \text{active terms are productive}\\
\text{return SAT at white fixpoint} && \text{decide or strictly refine}\\
\end{array}
\]

The revised solver still may remove exact duplicates inside the current batch. The theorem eliminates only subsumption against the old white snapshot; batch-local redundancy is an implementation issue.

\section{Revised Algorithm}
\label{sec:algorithm}

The snapshot/full-DNF solver below is the reference formulation of the revised algorithm and the revised configuration evaluated in Study~1:
\begin{quote}\small
\begin{enumerate}
\item Exhaustively apply unit propagation. Then return SAT for the empty network, return UNSAT if \((\Source,\Sink)\) is present, and return SAT if no non-empty white reach remains.
\item Freeze a snapshot \(\RN_t\), its white tails \(\W_t\), black reaches, and mixed reaches. Compute \(\T_t=\DNF(\W_t)\).
\item Initialize \(\Delta_t=\emptyset\). For each \(C\in\T_t\): if \(C\) is terminal, return SAT; if \(C\) is black-blocked, continue; otherwise find a \(C\)-falsified mixed reach, construct \(B_{\new}\) by~\eqref{eq:bnew}, and add it to \(\Delta_t\) without an old-white subsumption check.
\item If every \(C\in\T_t\) was black-blocked, return UNSAT. Otherwise \(\Delta_t\ne\emptyset\); add all tails in \(\Delta_t\) as white reaches and start the next iteration.
\end{enumerate}
\end{quote}
Soundness follows from unit propagation and private-pivot soundness. Termination follows because a non-decision iteration adds at least one old-white-non-subsumed white tail over the finite universe \(U\). Exact duplicate filtering may be used inside \(\Delta_t\); it does not replace the theoretical non-subsumption guarantee.

Two alternative DNF-control configurations are also implemented. Exact incremental DNF updates the previous complete DNF, with full recomputation when its update preconditions fail; antichain minimization may still be costly. Depth-first restart selects one current minimal token distribution, refines immediately, and restarts the search instead of materializing the complete current DNF in advance. Study~2 (Section~\ref{sec:study2-dnf-control}) compares their runtime and memory trade-offs. Snapshot/full-DNF remains the reference formulation; the measured ranking of implemented configurations is a separate empirical question.

\section{Empirical Evaluation}
\label{sec:empirical-evaluation}

The empirical evaluation is organized into four studies on SATLIB benchmarks~\cite{SATLIB}. Study~1 compares the original RN-Solver control flow with the revised snapshot/full-DNF subsumption-free variant. Study~2 compares the three implemented DNF-control strategies. Study~3 explores resolution-derived black enrichment, and Study~4 examines structural correlates of RN-Solver difficulty. The evaluation is not intended as a comparison with modern CDCL solvers: at this early stage, such a baseline on these small instances would primarily measure engineering maturity rather than the effect of the RN-specific refinement.

\subsection{Study 1: Original vs. Subsumption-Free Control Flow}
\label{sec:study1-control-flow}

Within the same implementation, we compare \texttt{old} and \texttt{sf-full-dnf} to evaluate the combined control-flow refinement: snapshot updates, learning from \(C\)-falsified mixed reaches, and removal of old-white subsumption checks. The counters establish elimination of the targeted check; runtime differences reflect the combined changes, not an isolated deletion.

The original instance-level measurements below provide a preliminary instrumentation check: one run per configuration, no warm-up, and an 8 second timeout. Their runtime medians include timeout-limit values; the subsequent full-benchmark study provides broader supporting evidence.

\subsubsection{Detailed Ten-Instance Comparison}
\label{sec:study1-detailed}

\paragraph{Original random 3-SAT.}
Table~\ref{tab:uf20-original} reports the first ten SATLIB \texttt{uf20-91} instances in natural filename order, from \texttt{uf20-01} through \texttt{uf20-010}; the full 1,000-instance benchmark is evaluated separately in Section~\ref{sec:uf20-full}. The revised solver eliminates the target check in every completed subsumption-free row: \[ \texttt{oldWhiteSubsumptionChecks}=0. \]

\begin{table}[!ht]
\centering
\caption{Original \texttt{uf20} subset. Times are milliseconds; 8000 denotes timeout. Timeout counters are left unreported (--) rather than interpreted as measured zeros.}
\label{tab:uf20-original}
\begin{tabular}{lrrrrrrr}
\hline
instance & white & black & mixed & old & sf-full & old checks & sf checks \\
\hline
\texttt{uf20-01}  & 10 & 11 & 70 & 921  & 261  & 343098 & 0 \\
\texttt{uf20-02}  & 11 & 13 & 67 & 24   & 21   & 341    & 0 \\
\texttt{uf20-03}  & 8  & 7  & 76 & 8000 & 8000 & -- & -- \\
\texttt{uf20-04}  & 11 & 14 & 66 & 1178 & 179  & 516279 & 0 \\
\texttt{uf20-05}  & 12 & 12 & 67 & 1134 & 109  & 404076 & 0 \\
\texttt{uf20-06}  & 11 & 8  & 72 & 8000 & 2675 & -- & 0 \\
\texttt{uf20-07}  & 13 & 12 & 66 & 327  & 217  & 53653  & 0 \\
\texttt{uf20-08}  & 11 & 10 & 70 & 29   & 22   & 1317   & 0 \\
\texttt{uf20-09}  & 14 & 15 & 62 & 8000 & 771  & -- & 0 \\
\texttt{uf20-010} & 18 & 9  & 64 & 8000 & 850  & -- & 0 \\
\hline
\end{tabular}
\end{table}

Under this reduced timeout protocol, the ratio of the median runtimes (old/sf) on the original \texttt{uf20} subset is about \(4.84\times\). The old solver performs many old-white subsumption checks on completed hard rows; the revised solver performs none.

\paragraph{Structured UNSAT sanity check.}
As a structured UNSAT baseline, we also ran \texttt{hole6.cnf} through \texttt{hole10.cnf}. The revised mode behaves similarly to the old mode on these instances: \texttt{hole6}--\texttt{hole9} remain in the same order of magnitude, and \texttt{hole10} times out in both modes under the 8 second cap. This baseline is useful for regression, but the main claim of this paper is supported more directly by the random 3-SAT subsumption counters.

\begin{table}[t]
\centering
\caption{Reported counters in the original instrumented batch. The retained summary does not establish whether debug-only checks were enabled; zero reported violations are not evidence of active lemma verification.}
\label{tab:invariant-checks}
\begin{tabular}{lrrrr}
\hline
mode & instances & complement violations & private witness violations & active resolved \\
\hline
\texttt{old}         & 45 & 0 & 0 & 0 \\
\texttt{sf-full-dnf} & 45 & 0 & 0 & 0 \\
\hline
\end{tabular}
\end{table}

The detailed measurements support the main implementation claim: the revised solver eliminates the targeted old-white subsumption check.

\subsubsection{Full \texttt{uf20-91} Fixed-Budget Validation}
\label{sec:uf20-full}

To test whether the instance-level pattern extends beyond the ten-instance table, we ran a supporting fixed-budget ablation on all 1,000 SATLIB \texttt{uf20-91} instances. Both conditions used the same frozen instrumented implementation: \texttt{old} for the original control flow and \texttt{sf-full-dnf} for the revised snapshot/full-DNF control flow. This campaign ran on an Intel Core i7-1255U with 32~GiB RAM, Windows~11, and OpenJDK~21.0.1. Each instance--mode pair was executed once in a fresh JVM, serially, with the implementation's MMCS-style builder, the neutral heuristic, identical heap settings (\texttt{-Xms64m -Xmx1g}), and an external 8 second deadline for the complete Java process. The mode executed first was alternated between successive instances. This supports the control-flow comparison rather than replacing the detailed experiment.

\begin{table}[!ht]
\centering
\small
\caption{Supporting fixed-budget validation on all 1,000 \texttt{uf20-91} instances. Runtime statistics include successful runs only.}
\label{tab:uf20-full}
\begin{tabular}{lrrrr}
\hline
mode & solved & timeouts & mean ms & median ms \\
\hline
\texttt{old}         & 724/1000 & 276 & 1500.6 & 731.8 \\
\texttt{sf-full-dnf} & 963/1000 & 37  & 851.3  & 430.9 \\
\hline
\end{tabular}
\end{table}

Under the same fixed budget, \texttt{sf-full-dnf} solved 963 instances (96.3\%), compared with 724 (72.4\%) for \texttt{old}, a difference of 23.9 percentage points. The paired outcomes were 724 solved by both modes, 0 by \texttt{old} only, 239 by \texttt{sf-full-dnf} only, and 37 by neither; a supplementary two-sided exact McNemar test gives \(p=2.26\times 10^{-72}\). The paired counts are the primary descriptive result for this fixed benchmark. On the 724 instances solved by both modes, the median runtime decreased from 731.8~ms to 356.1~ms. The median per-instance old/sf runtime ratio was \(1.783\times\), and its geometric mean was \(2.180\times\). Because the independently successful sets differ, their separate means in Table~\ref{tab:uf20-full} are descriptive rather than a direct paired speed comparison.

All 1,687 successful SAT models produced across the two modes were independently checked against the 91 input clauses and were valid, with no paired SAT/UNSAT disagreement. Each of the 963 completed \texttt{sf-full-dnf} runs reported zero old-white subsumption checks. Peak working-set memory was measured externally. On the same 724 instances solved by both modes, the median peak working set was 175.3~MiB for \texttt{old} and 98.5~MiB for \texttt{sf-full-dnf} (means 174.3 and 121.8~MiB, respectively); these values measure the peak working set of the complete JVM process, not Java heap allocation. The experiment used one cold-JVM measurement per instance--mode pair, so the runtime includes process/JVM overhead and does not estimate repeated-run variance. These results support, rather than replace, the more detailed instance-level experiment above.

\subsection{Study 2: DNF-Control Strategy Comparison}
\label{sec:study2-dnf-control}

We compared the three implemented configurations, snapshot/full-DNF (\texttt{sf-full-dnf}), exact incremental DNF (\texttt{sf-inc-dnf}), and depth-first restart (\texttt{sf-depth-restart}), on the next 100 \texttt{uf20-91} instances after a ten-instance strategy pilot. The frozen holdout comprises natural filename positions 11--110, from \texttt{uf20-011.cnf} to \texttt{uf20-0110.cnf}; strategy-pilot measurements were not pooled with it. This holdout is disjoint from that pilot, not from the 1,000-instance set in Study~1. Each instance--configuration pair was run three times in a fresh JVM, giving 900 serial executions on Windows~11 with OpenJDK~21.0.1. All runs used the same frozen v4 build, \texttt{-Xms64m -Xmx1g}, the neutral heuristic, no debug checks or batch minimization, and an external 8 second full-process deadline. Mode order rotated so that each configuration occupied each position once per instance. Full enumeration and incremental initialization/fallback used the implementation's MMCS-style builder; depth-restart used its own DFS oracle despite accepting the same dualization option. Incremental retained the 1,000,000-candidate limit with full-DNF fallback enabled. These are implemented-configuration comparisons, not a pure traversal ablation: enumeration order, duplicate handling, private-witness selection, and batch versus immediate learning also differ.

\begin{table}[!ht]
\centering
\small
\setlength{\tabcolsep}{5pt}
\renewcommand{\arraystretch}{1.12}
\caption{Study~2 on 100 instances with three repetitions per configuration. Instance counts require completion in all three repetitions; run counts include all repetitions. Capped wall time and peak working set use all 300 runs per configuration, including timeouts; solver time uses completed runs only. All time and memory entries are medians.}
\label{tab:dnf-control-strategies}
\begin{tabular}{lrrrrr}
\hline
\noalign{\vskip 3pt}
Configuration & \shortstack[r]{Instances\\(3/3)} & \shortstack[r]{Completed\\runs} & \shortstack[r]{Capped wall\\(ms)} & \shortstack[r]{Solver\\(ms)} & \shortstack[r]{Peak\\(MiB)} \\[2pt]
\hline
Full-DNF      & 94/100  & 282/300 & 545.1  & 186.0 & 115.4 \\
Incremental   & 100/100 & 300/300 & 403.8  & 120.0 & 67.5  \\
Depth-restart & 60/100  & 187/300 & 4438.5 & 752.0 & 89.9  \\
\hline
\end{tabular}
\end{table}

Table~\ref{tab:dnf-control-strategies} reports completion and resource use. The primary problem units are the 100 instances, not the 300 repeated runs per configuration. Capped wall time is \(\min(t_{\mathrm{wall}},8000\text{ ms})\), not an uncensored solve-time estimate. Peak memory is the externally measured lifetime Windows \texttt{PeakWorkingSetSize} of the complete JVM, including observations terminated at the deadline; it is not Java heap allocation. Incremental completed all repetitions of all 100 instances, whereas full-DNF did so for 94 and depth-restart for 60. Depth-restart completed at least once on 64 instances and timed out in 113 of its 300 runs. All 769 completed runs returned independently validated SAT models, with no UNSAT answer, process error, or out-of-memory error. Three additional valid models emitted by processes that exceeded the deadline remained classified as timeouts.

On the 282 full/incremental paired completions across 94 instances, the median incremental/full wall-time ratio was 0.807 (95\% CI \([0.767,0.855]\)); the corresponding median internal solver-time ratio was 0.665. The median paired peak-working-set ratio was 0.597 (95\% CI \([0.498,0.623]\)). Values below one favor incremental. These percentile intervals use 2,000 instance-cluster bootstrap resamples~\cite{FieldWelsh2007}, retaining eligible paired repetitions together; the ratios are conditional on both configurations completing. Genuine incremental updates occurred in 276/300 runs (92\%), and no full-DNF fallback occurred. The remaining 24 runs required only initial full enumeration. Thus the observed incremental advantage was not an artifact of frequent fallback. Full-DNF also had a heavier memory tail: its maximum observed peak working set was 642.1~MiB, compared with 114.5~MiB for incremental.

Depth-restart remained less effective under this budget. On its 187 paired completions with full-DNF, the median depth/full wall-time ratio was 1.788. Oracle search accounted for a median 90.6\% of measured solver time in the 187 completed depth runs, exceeding half the solver time in 161 of them; no phase allocation is inferred for timeouts. Its median paired peak-working-set ratio was 0.841 relative to full-DNF but 1.354 relative to incremental, so depth-restart was not the lowest-memory configuration. These supporting results identify incremental DNF as the most effective configuration on this holdout while retaining snapshot/full-DNF as the reference formulation of the revised algorithm. The findings concern a deterministic, contiguous set of small satisfiable instances; they do not establish a universal strategy ranking, and completed-only timing summaries describe different selected subsets.

\subsection{Study 3: Exploratory Resolution-Derived Black Enrichment}
\label{sec:study3-black-enrichment}

A white-enrichment control experiment with random positive clauses did not make random 3-SAT easier and is omitted from the main tables. A more relevant exploratory experiment enriches the black side, but only with clauses derived from the original formula by bounded resolution. Thus every added all-negative clause has a recorded resolution derivation and preserves the original SAT/UNSAT status.

\begin{table}[!ht]
\centering
\caption{Resolution-derived black enrichment on the reduced \texttt{uf20} subset.}
\label{tab:derived-black}
\begin{tabular}{lrrrrrrr}
\hline
 target & reached & ratio & added & deriv. ms & old & sf-full & sf DNF terms \\
\hline
 original & 10/10 & 0.126 & 0 & 0  & 664 & 138 & 3988 \\
 5\%      & 10/10 & 0.126 & 0 & 23 & --  & --  & --   \\
 10\%     & 10/10 & 0.126 & 0 & 23 & --  & --  & --   \\
 15\%     & 10/10 & 0.153 & 3 & 38 & --  & --  & --   \\
 20\%     & 10/10 & 0.202 & 8 & 44 & 454 & 136 & 2555 \\
\hline
\end{tabular}
\end{table}

The result is suggestive but not conclusive. The derived black clauses reduce median materialized DNF terms in \texttt{sf-full-dnf} from 3988 to 2555 at the 20\% target, and the old solver improves from 664 ms to 454 ms. However, the measured \texttt{blackBlockedCount} and \texttt{dnfBranchesPrunedByBlack} counters remain at zero in the reduced satisfiable benchmark set. Thus the current data do not justify the stronger claim that the improvement is caused by increased black-aware pruning. The effect may occur through the shape of the DNF search or through other simplifications; larger runs and finer instrumentation are needed.

\subsection{Study 4: Structural Correlates of Solver Difficulty}
\label{sec:study4-structural-correlates}

As a post-hoc supporting analysis, we applied unchanged \texttt{CnfStats} to all 1,000 formulas and joined the clause-sign counts to Study~1 outcomes; no solver was rerun. Main inference excluded the ten hypothesis-generating pilot instances (\(n=990\)); this is not an independent replication. We used two-sided, tie-corrected Mann--Whitney tests with Holm correction~\cite{Holm1979} across six prespecified features within each of three comparison families (success versus timeout separately by mode, and SF-rescue versus both-timeout). The features were white and black counts, their sum, their absolute difference, definite-Horn count, and positive-literal count. Negative Cliff's \(\delta\)~\cite{Cliff1993} denotes lower values in the timeout group. The clearest association was black-clause count: at 8 seconds, the successful/timeout medians were 12/9 for \texttt{old} (\(\delta=-0.438\), Holm-adjusted \(p=5.0\times10^{-26}\)) and 11/8 for \texttt{sf-full-dnf} (\(\delta=-0.606\), \(p=3.33\times10^{-9}\)). SF-rescued cases likewise had more black clauses than cases timing out in both modes (medians 10/8; \(\delta=-0.359\), \(p=0.00285\)). Logistic models used three nonredundant standardized counts: white, black, and definite-Horn. At fixed white and definite-Horn counts and 91 clauses, increasing black count replaces clauses with two positive and one negative literal. A one-standard-deviation increase had timeout odds ratios of 0.381 for \texttt{old} (Wald 95\% CI \([0.315,0.461]\)) and 0.237 for \texttt{sf-full-dnf} (\([0.147,0.379]\)); the adjusted white-count association was inconclusive. Only 36 holdout timeouts informed the latter model.

Sensitivity analysis reclassified the same full-process runtimes at 4, 2, and 1 seconds: these are dependent recensorings, not new solver experiments. The association of fewer black clauses with timeout persisted for both modes and the rescue comparison, including Holm adjustment across all 72 feature/comparison/deadline tests, although effect sizes generally decreased at shorter deadlines. These findings are exploratory and non-causal. \texttt{CnfStats} measures clause-sign composition, not incidence topology, white-tail overlap, hitting-set counts, DNF size, or private-witness structure. Black-clause abundance is therefore a correlate of observed RN-Solver difficulty, not evidence that increasing it necessarily causes faster solving.

\section{Conclusion}
\label{sec:conclusion}

We presented a subsumption-free refinement of RN-Solver. The key result is that when private-pivot learning is restricted to \(C\)-falsified mixed reaches, the generated white tail is disjoint from \(C\), while every old white tail intersects \(C\). Therefore, no old white reach can subsume the generated candidate, and a runtime subsumption check is replaced by a structural theorem.

The revised snapshot/full-DNF solver also simplifies the logical control flow. In a frozen snapshot, every active token distribution is productive: if it is not terminal and not black-blocked, it yields a \(C\)-falsified mixed reach and hence a non-subsumed private-pivot refinement. The old resolved marker and white-fixpoint SAT branch are therefore unnecessary in the revised algorithm. Each non-decision iteration performs a strict white refinement and eliminates the token distribution that generated the new reach.

Study~1 confirms the intended elimination of old-white subsumption checks and shows an increase in 8 second coverage from 724 to 963 of 1,000 \texttt{uf20-91} instances. Study~2 identifies exact incremental DNF as the most effective tested configuration on its 100-instance strategy holdout: it completed all three repetitions of every instance and reduced paired runtime and peak working set relative to full-DNF. Snapshot/full-DNF remains the reference formulation; the strategy ranking is empirical and benchmark-specific. Studies~3 and~4 provide exploratory evidence about black-side enrichment and clause-sign associations, not a general causal explanation. RN-Solver remains a proof-of-concept rather than a competitor to optimized CDCL solvers.

Future work has four immediate directions. First, the resolution-derived black-enrichment experiment should be repeated on larger SAT and UNSAT benchmarks with finer instrumentation. Its preliminary reduction in DNF materialization suggests that derived black consequences can influence the search space, even when the current black-pruning counters remain flat. Second, this observation motivates a genuinely bidirectional RN-Solver, where black reaches are generated dynamically from the Sink side, symmetrically to the current generation of white reaches from Source. Third, the incremental advantage should be evaluated on larger, more diverse SAT and UNSAT families, while depth-first and unit-oriented DNF control require more efficient oracles that certify minimality without shifting the cost from materialization to search.

A fourth direction is polarity renaming: globally complement selected variables, seeking to increase the initial number of black reaches. This preserves satisfiability through a bijection of assignments but changes clause-sign composition. Study~4 motivates, but does not establish, a performance benefit. A dedicated paired comparison should keep solver settings fixed, include preprocessing cost, map models back to the original variables, and account for the other clause-sign changes induced by renaming.

\bibliographystyle{eptcs}
\bibliography{rn_subsumption_free_draft_v3_9}

\end{document}